\documentclass[notitlepage,a4,11pt]{article}

\usepackage{graphicx}

\usepackage{amsmath}%
\usepackage{amsfonts}%
\usepackage{amssymb}%
\usepackage{mathrsfs}
\usepackage{amsthm}

\usepackage{authblk}

\newcommand{\cp}{\times}

\newcommand{\cu}{\nabla\times} 
\newcommand{\JI}[1]{\bol{#1}\cdot\cu{\bol{#1}}}

\newcommand{\bol}{\boldsymbol}
\newcommand{\hb}[1]{\hat{\bol{#1}}}
\newcommand{\abs}[1]{\left\lvert{#1}\right\rvert}
\newcommand{\w}{\wedge}
\newcommand{\lr}[1]{\left({#1}\right)}

\newcommand{\mf}{\mathfrak}
\newcommand{\p}{\partial}

\newcommand{\ti}[1]{\textit{#1}}
\newcommand{\tb}[1]{\textbf{#1}}
\newcommand{\ov}[1]{\mkern 1.5mu\overline{\mkern-1.5mu#1\mkern-1.5mu}\mkern 1.5mu}

\newtheorem{theorem}{\textit{Theorem}}
\newtheorem{corollary}{\textit{Corollary}}
\newtheorem{proposition}{\textit{Proposition}}

\begin{document}
\title{Local Representation and Construction\protect\\of Beltrami Fields}
\author[1]{N. Sato} \author[1]{M. Yamada}
\affil[1]{Research Institute for Mathematical Sciences, \protect\\ Kyoto University, Kyoto 606-8502, Japan \protect\\ Email: sato@kurims.kyoto-u.ac.jp}
\date{\today}
\setcounter{Maxaffil}{0}
\renewcommand\Affilfont{\itshape\small}


\maketitle

\begin{abstract}
A Beltrami field is an eigenvector of the curl operator.
Beltrami fields describe steady flows in fluid dynamics and force free magnetic fields in plasma turbulence. 
By application of the Lie-Darboux theorem of differential geoemtry, 
we prove a local representation theorem for Beltrami fields.
We find that, locally, a Beltrami field has a standard form 
amenable to an Arnold-Beltrami-Childress flow with two of the parameters set to zero.
Furthermore, a Beltrami flow admits two local invariants, 
a coordinate representing the physical plane of the flow, and
an angular momentum-like quantity in the direction across the plane. 
As a consequence of the theorem, we derive a method to construct
Beltrami fields with given proportionality factor.
This method, based on the solution of the eikonal equation,
guarantees the existence of Beltrami fields for any 
orthogonal coordinate system such that at least two scale factors are equal.
We construct several solenoidal and non-solenoidal Beltrami fields with both homogeneous and inhomogeneous proportionality factors.
\end{abstract}





\section{Introduction}
Let $\bol{w}\in C^{\infty}\lr{{\Omega}}$ be a smooth vector field in a bounded domain $\Omega\subset\mathbb{R}^3$. Object of the present study is the following system of first order partial differential equations:
\begin{equation}
\bol{w}\cp\lr{\nabla\cp\bol{w}}=\bol{0}~~~~{\rm in}~~\Omega.\label{Beq}
\end{equation}
A solution $\bol{w}$ to \eqref{Beq} is called a Beltrami field.
Evidently, a Beltrami field $\bol{w}\neq\bol{0}$ in $\Omega$ 
is an eigenvector of the curl operator \cite{YRot}, i.e. it satisfies:
\begin{equation}
\cu{\bol{w}}=\alpha\,\bol{w}~~~~{\rm in}~~\Omega,\label{Beq2}
\end{equation}
where $\alpha\in C^{\infty}\lr{\Omega}$ is the proportionality factor (eigenvalue).

Beltrami fields arise as stationary solutions of the Euler equations in fluid dynamics \cite{Moffatt14,Moffatt85,Peralta} and as force free magnetic fields in magnetohydrodynamics \cite{Woltjer,Yos2002,Mah}. Indeed, the steady ideal Euler equations at constant density, 
\begin{equation}\label{IE}
\lr{\bol{w}\cdot\nabla}\bol{w}=-\nabla P,~~~~\nabla\cdot\bol{w}=0~~~~{\rm in}~~\Omega,
\end{equation} 
reduce to the equation for a solenoidal Beltrami field,
\begin{equation}\label{IE2}
\bol{w}\cp\lr{\cu{\bol{w}}}=\bol{0},~~~~\nabla\cdot\bol{w}=0~~~~{\rm in}~~\Omega,
\end{equation} 
whenever the pressure $P$ is given by $P=-w^{2}/2$. 
\eqref{IE2} is also the equation satisfied by a force free magnetic field.
Beltrami fields occur in topologically constrained systems as well, 
where they are operators acting on a Hamiltonian function to generate
particle dynamics\footnote{If $\bol{w}$ is a Beltrami operator, particle dynamics obeys the equation of motion $\dot{\bol{x}}=\bol{w}\cp\nabla H$, with $\bol{x}$ the particle position in $\mathbb{R}^{3}$ and $H$ a scalar function (the Hamiltonian).}\cite{Sato18}.  

In addition to physical applications, 
Beltrami fields play a key role in understanding 
topological properties of steady Euler flows \cite{Etnyre}.
In Refs. \cite{Enciso2,Enciso3,Enciso4} it is shown that
vortex tubes of arbitrarily complex topology can be
realized by means of Beltrami fields.
Despite the centrality of Beltrami fields in the study of fluid flows and plasma turbulence \cite{Dombre,Taylor}, there are
some aspects pertaining to their geometrical properties
that have not yet been clarified. 
This can be understood by comparison with complex lamellar vector fields, 
i.e. vector fields $\bol{w}\in C^{\infty}\lr{\Omega}$ with vanishing helicity density:
\begin{equation}
h=\JI{w}=0~~~~{\rm in}~~\Omega.\label{FIC}
\end{equation}
This condition is opposite to equation \eqref{Beq}, in the sense that
while \eqref{Beq} requires alignment between vector field and curl,
equation \eqref{FIC} requires orthogonality.
It is known (Frobenius theorem, see Ref. \cite{Frankel}) that
a vector field obeying \eqref{FIC} is \textit{integrable}.
More precisely, for any $\bol{x}\in\Omega$, there exists a neighborhood
$U\subset\Omega$ of $\bol{x}$ such that
\begin{equation}
\bol{w}=\lambda\,\nabla C~~~~{\rm in}~~U,
\end{equation} 
where $\lr{\lambda,C}\in C^{\infty}\lr{U}$ are smooth functions in $U$.
However, no local representation theorem is known for Beltrami fields.
Furthermore, exception made for standard Arnlod-Beltrami-Childress (ABC) flows \cite{Dombre,Zhao},
\begin{equation}
\begin{split}
\bol{w}=&\lr{A \sin{z}+C \cos{y}}\nabla x+\lr{B \sin{x}+A \cos{z}}\nabla y\\
        &+\lr{C\sin{y}+B\cos{x}}\nabla z,~~~~A,B,C\in\mathbb{R},\label{ABC}
\end{split}
\end{equation} 
very few explicit examples of Beltrami fields are known, most of them with constant proportionality factors, i.e. such that $\nabla\cp\bol{w}=\alpha\,\bol{w}$ with $\alpha\neq 0$ a real constant.
This problem is both intrinsic and technical: on one hand
the solution of \eqref{Beq} for general inhomogeneous proportionality factors
is locally obstructed by the inhomogeneity, leading to Beltrami fields with low-regularity \cite{Enciso,Kaiser,Kress}. 
On the other hand any Beltrami field with nonzero proportionality factor is inevitably nonintegrable (it cannot satisfy $h=0$). 
This automatically complicates the topology of the vector field, hiding the `natural' form of the solution to \eqref{Beq}.  
In terms of the Clebsch parametrization of a vector field \cite{YClebsch,Yos2017}, this means that the representation of a Beltrami field with nonzero proportionality factor will always need more than two Clebsch parameters. Hence, Beltrami fields with nonzero proportionality factors are fully three-dimensional objects,
and can be related to the notion of Reeb vector field in contact topology \cite{Etnyre2}.

This paper is organized as follows.
First, we prove the following local representation theorem for Beltrami fields by applying the Lie-Darboux theorem of differential geometry \cite{DeLeon_3,Arnold_4,Salmon}:
\begin{theorem}\label{LocRep}
Let $\bol{w}\in C^{\infty}\lr{\Omega}$ be a smooth vector field in a bounded domain $\Omega\subset\mathbb{R}^3$ with $h=\bol{w}\cdot\nabla\cp\bol{w}\neq\bol{0}$ in $\Omega$. Then $\bol{w}$ satisfies equation \eqref{Beq} if and only if for every $\bol{x}\in\Omega$ there exists a neighborhood $U\subset\Omega$ of $\bol{x}$ and a local coordinate system $\lr{\ell,\psi,\theta}\in C^{\infty}\lr{U}$ such that
\begin{subequations}\label{eq2}
\begin{align}
&\cos{\theta}~\sin{\theta}\lr{\abs{\nabla\psi}^2-\abs{\nabla\ell}^2}=\lr{\nabla\ell\cdot\nabla\psi}\lr{\cos^2{\theta}-\sin^2{\theta}},\\
&\sin{\theta}~\nabla\ell\cdot\nabla\theta+\cos{\theta}~\nabla\psi\cdot\nabla\theta=0,\label{2b}
\end{align}
\end{subequations}
and
\begin{equation}\label{eq3}
\bol{w}=\cos{\theta}~\nabla\psi+\sin{\theta}~\nabla\ell~~~~{\rm in}~~U.
\end{equation}
\end{theorem}
A direct consequence of Theorem \ref{LocRep} is that
the flow generated by a Beltrami field admits two local invariants,
one representing the physical plane of the flow, the other conservation
of an angular momentum-like quantity in the direction across the plane.
The existence of two local invariants suggests a description of Beltrami flows in terms of Nambu brackets \cite{Nambu}.
We also find that, if the Beltrami field is solenoidal, 
the proportionality factor is a function of the local invariants, 
giving an explanation of the helical flow paradox \cite{Morgulis}.
In addition, Theorem \ref{LocRep} naturally leads to the following method for the construction of Beltrami fields with given proportionality factor $\alpha$, i.e. such that $\nabla\cp\bol{w}=\alpha\,\bol{w}$.
\begin{corollary}\label{Construction}
Let $\lr{\ell,\psi,\theta}\in C^{1}\lr{\Omega}$ be an orthogonal coordinate system in a bounded domain $\Omega\in\mathbb{R}^3$ such that, in $\Omega$,
\begin{subequations}
\begin{align}
\abs{\nabla\theta}&=\lvert\alpha\rvert,\label{Eik}\\
\abs{\nabla\ell}&=\abs{\nabla\psi},
\end{align}
\end{subequations} 
where $\alpha\in C\lr{\Omega}$. Then, the vector fields
\begin{subequations}
\begin{align}
&\bol{w}=\cos{\theta}~\nabla\psi+\sin{\theta}~\nabla\ell,\\
&\bol{w}^{\ast}=\sin{\theta}~\nabla\psi+\cos{\theta}~\nabla\ell,
\end{align}
\end{subequations}
are Beltrami fields in $\Omega$ with proportionality factors $\sigma\lvert\alpha\rvert$ and $-\sigma\lvert\alpha\rvert$ respectively, i.e. $\nabla\cp\bol{w}=\sigma\lvert\alpha\rvert\,\bol{w}$ and $\nabla\cp\bol{w}^{\ast}=-\sigma\lvert\alpha\rvert\,\bol{w}^{\ast}$, where $\sigma=h/\abs{h}$ is the sign of the helicity density $h=\JI{w}$. 
\end{corollary}
In the second part of the paper we apply corollary \ref{Construction} to
construct Beltrami fields with homogeneous and inhomogeneous proportionality factors.
Both solenoidal and non-solenoidal examples will be given.

\section{Preliminaries}
Consider equation \eqref{Beq2}.
If $\alpha$ is a nonzero real constant, $\bol{w}$ is called a strong (or linear) Beltrami field. 
The helicity density $h$ of the vector field $\bol{w}$ can be evalueated as
\begin{equation}
h=\JI{w}=\alpha~{w}^2.\label{h}
\end{equation}
Here $w=\abs{\bol{w}}$. We shall say that a Beltrami field is \ti{nontrivial}
whenever $h\neq 0$ in $\Omega$. If $h\neq 0$, we have $w\neq 0$ as well as $\cu{\bol{w}}\neq\bol{0}$ in $\Omega$. Hence, the proportionality factor $\alpha$ can be expressed in terms of the helicity density $h$ as
\begin{equation}
\alpha=w^{-2}\,h=\hat{h}.
\end{equation}
Here $\hat{h}=w^{-2}\,h=\hb{w}\cdot\nabla\cp\hb{w}$ represents the helicity density of the normalized vector field $\hb{w}=\bol{w}/w$.
In the following we will always be concerned with nontrivial Beltrami fields and adopt the notation
\begin{equation}
\cu{\bol{w}}=\hat{h}~\bol{w}.\label{BF}
\end{equation}
If $\bol{w}$ is solenoidal, i.e. $\nabla\cdot\bol{w}=0$, the helicity density $\hat{h}$ is an integral invariant of the flow generated by $\bol{w}$,
\begin{equation}
\nabla\hat{h}\cdot\bol{w}=0.
\end{equation} 
This can be verified by taking the divergence of equation \eqref{BF}.

The simplest example of nontrivial Beltrami field in $\mathbb{R}^3$ is
the vector field
\begin{equation}
\bol{w}=\sin{z}~\nabla x+\cos z~\nabla y,\label{B0}
\end{equation}
which satisfies $w=1$, $\cu{\bol{w}}=\bol{w}$, $\hat{h}=1$, and $\nabla\cdot\bol{w}=0$. Furthermore, \eqref{B0} can be written as
\begin{equation}
\bol{w}=\nabla z\cp\nabla\lr{x\cos{z}-y\sin{z}}.
\end{equation} 
Hence the flow generated by \eqref{B0} has two integral invariants: the variable $z$ which labels the planes where the flow lies, and the $z$-component of the angular momentum $\bol{L}=\bol{x}\cp\bol{w}$, $L_{z}=x\cos{z}-y\sin{z}$. Figure \ref{fig1} shows the behavior of the Beltrami field \eqref{B0} on the level sets $z=\rm constant$ and $L_{z}=\rm constant$. While the Beltrami field \eqref{B0} is very simple, 
it encloses all local properties of general Beltrami fields. 
This is the content of Theorem \ref{LocRep}.

\begin{figure}[h]
\hspace*{-0cm}\centering
\includegraphics[scale=0.3]{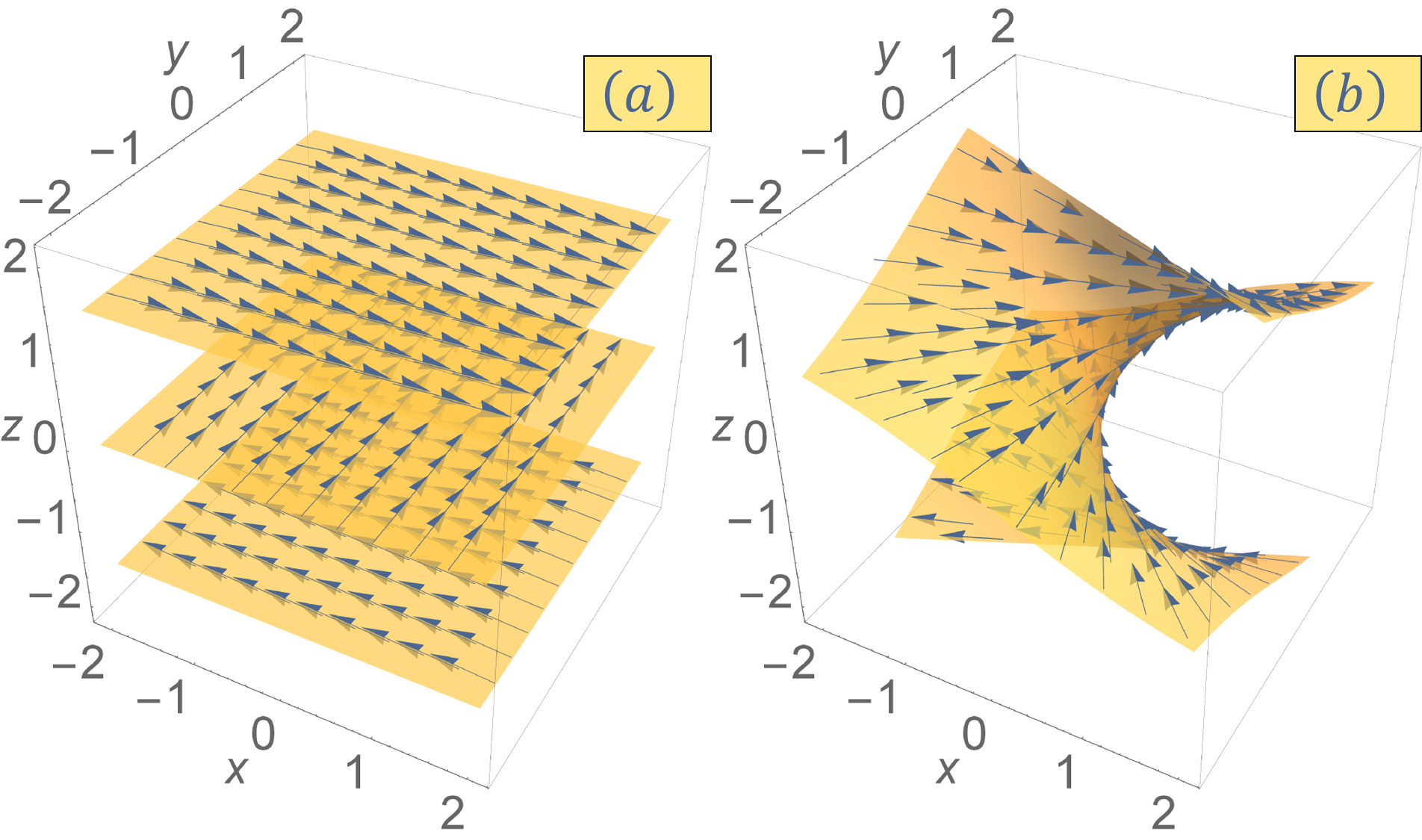}
\caption{\footnotesize (a): Plot of \eqref{B0} on the surfaces $z=-1.5,0,1.5$. (b): Plot of \eqref{B0} on the surface $L_{z}=0$.}
\label{fig1}
\end{figure}

It is useful to make some considerations on the relation between 
equation \eqref{Beq}, and the Frobenius integrability condition \eqref{FIC} for the vector field $\bol{w}$. We have already seen that, when \eqref{FIC} holds, there exists local smooth functions $\lambda$ and $C$ such that $\bol{w}=\lambda\,\nabla C$. Evidently, a Beltrami field can never satisfy \eqref{FIC} unless $\alpha\,w^2=0$. Hence, either $\alpha=0$ or $w=0$ in $\Omega$. Both imply $\cu{\bol{w}}=\bol{0}$ in $\Omega$. 
Thus, when $\cu{\bol{w}}\neq\bol{0}$, the Frobenius integrability condition \eqref{FIC} is `dual' to the Beltrami field condition \eqref{Beq}, in the sense that one requires orthogonality between vector field and curl, the other their alignment.

\section{Local representation of Beltrami fields}
In this section we prove Theorem \ref{LocRep}.

\begin{proof}
First we prove that \eqref{Beq} implies equations \eqref{eq2} and \eqref{eq3}.
Let ${\rm vol}=dx\w dy\w dz$ be the standard volume form of $\mathbb{R}^3$.
To the vector field $\bol{w}=\lr{w_{x},w_{y},w_{z}}$ we associate the 1-form
\begin{equation}
w^{1}=\ast\,i_{\bol{w}}{\rm vol}=w_{x}dx+w_{y}dy+w_{z}dz.
\end{equation}
Here $i$ is the contraction operator and $\ast$ the Hodge star operator
defined with respect to the Euclidean metric of $\mathbb{R}^{3}$.
Next we define the 2-form $\omega$,
\begin{equation}
\begin{split}
\omega&=dw^{1}\\
&=\lr{\frac{\p w_{z}}{\p y}-\frac{\p w_{y}}{\p z}}\ast dx+
\lr{\frac{\p w_{x}}{\p z}-\frac{\p w_{z}}{\p x}}\ast dy+
\lr{\frac{\p w_{y}}{\p x}-\frac{\p w_{x}}{\p y}}\ast dz\\
&=\lr{\cu{\bol{w}}}^{i} \ast dx^{i}.\label{dw1}
\end{split}
\end{equation}
On the other hand, any 2-form $\omega$ can be expressed in terms of an antisymmetric matrix $\omega_{ij}=-\omega_{ji}$ as $\omega=\sum_{i<j}\omega_{ij}\,dx^{i}\w dx^{j}$. In $\mathbb{R}^{3}$ this gives $\omega=\omega_{yz}\ast dx+\omega_{zx}\ast dy+\omega_{xy}\ast dz$. Hence $\omega_{ij}=\epsilon_{ijk}\lr{\cu{\bol{w}}}^{k}$.

The rank of the 2-form $\omega$ is determined by the rank of the matrix $\omega_{ij}$.
Since by hypothesis $h\neq 0$ in $\Omega$, we have $\cu{\bol{w}}\neq\bol{0}$ in $\Omega$. 
Thus the 2-form $\omega$ has rank 2 in $\Omega$ (in 3 dimensions any non-vanishing antisymmetric matrix has rank 2).
Furthermore, the 2-form $\omega$ is closed:
\begin{equation}
d\omega=ddw^{1}=0~~~~{\rm in}~~\Omega.
\end{equation} 
Therefore the hypothesis of the Lie-Darboux theorem \cite{DeLeon_3,Arnold_4,Salmon} are verified. Hence, for every point $\bol{x}\in\Omega$ there exists
a neighborhood $U\subset\Omega$ of $\bol{x}$ and functions $\lr{\lambda,C}\in C^{\infty}\lr{U}$ such that
\begin{equation}
\omega=d\lambda\w dC~~~~{\rm in}~~U.
\end{equation}
This implies $\omega=\epsilon_{ijk}\lambda_{i}C_{j}\ast dx^{k}$ and, from \eqref{dw1},
\begin{equation}
\cu{\bol{w}}=\nabla\lambda\cp\nabla C~~~~{\rm in}~~U.\label{curlw}
\end{equation}
The curl in equation \eqref{curlw} can be removed to give
\begin{equation}
\bol{w}=\nabla\mu+\lambda\,\nabla C~~~~{\rm in}~~U.\label{w}
\end{equation}
Since both $\bol{w}$ and $\lambda\,\nabla C$ are of class $C^{\infty}\lr{U}$, we have $\mu\in C^{\infty}\lr{U}$.

Next recall that, by hypothesis, $\bol{w}$ 
is a nontrivial Beltrami field such that $h\neq 0$ in $\Omega$. On the other hand, from equation \eqref{w}, we have
\begin{equation}
h=\nabla\mu\cdot\nabla\lambda\cp\nabla C\neq 0~~~~{\rm in}~~U.
\end{equation}
Hence, $\lr{\mu,\lambda,C}$ is a coordinate system in $U$ such that the Jacobian of the transformation $\lr{\mu,\lambda,C}\mapsto\lr{x,y,z}$ is given by $h$. We define a second coordinate change:
\begin{subequations}
\begin{align}
\ell&=\mu \sin{C}-\lambda \cos{C},\\
\psi&=\mu \cos{C}+\lambda \sin{C},\\
\theta&=C.
\end{align}
\end{subequations}
The coordinates $\lr{\ell,\psi,\theta}$ are well defined in $U$ since the Jacobian of the coordinate change $\lr{\ell,\psi,\theta}\mapsto\lr{\mu,\lambda,C}$ is $1$, i.e. $h=\nabla\mu\cdot\nabla\lambda\cp\nabla C=\nabla\ell\cdot\nabla\psi\cp\nabla\theta$. The inverse transformation is:
\begin{subequations}\label{InvTransf}
\begin{align}
\mu&=\psi \cos{\theta}+\ell \sin{\theta},\\
\lambda&=\psi \sin{\theta}-\ell \cos{\theta},\\
C&=\theta.
\end{align}
\end{subequations}
Substituting \eqref{InvTransf} into equation \eqref{w}, we have
\begin{equation}
\bol{w}=\cos{\theta}\,\nabla\psi+\sin{\theta}\,\nabla\ell~~~~{\rm in}~~U.\label{w2}
\end{equation} 
Notice that, since the Beltrami field equation \eqref{Beq} has not yet been used,
the local representation \eqref{w2} holds for any vector field $\bol{w}$ such that $h\neq 0$ in $\Omega$. Rewriting equation \eqref{Beq} in terms of the expression for $\bol{w}$ obtained above gives
\begin{equation}\label{BF4}
\sin{\theta}~\nabla\psi\cp\nabla\theta+\cos{\theta}~\nabla\theta\cp\nabla\ell=\frac{\nabla\ell\cdot\nabla\psi\cp\nabla\theta}{w^2}\lr{\cos{\theta}~\nabla\psi+\sin{\theta}~\nabla\ell}.
\end{equation}
Next, recall that tangent basis vectors $\lr{\p_{\ell},\p_{\psi},\p_{\theta}}$ can be expressed as
\begin{equation}\label{tvectors}
\p_{\ell}=\frac{\nabla\psi\cp\nabla\theta}{\nabla\ell\cdot\nabla\psi\cp\nabla\theta},~~
\p_{\psi}=\frac{\nabla\theta\cp\nabla\ell}{\nabla\ell\cdot\nabla\psi\cp\nabla\theta},~~
\p_{\theta}=\frac{\nabla\ell\cp\nabla\psi}{\nabla\ell\cdot\nabla\psi\cp\nabla\theta}.
\end{equation}
Substituting these expressions in \eqref{BF4} we obtain
\begin{equation}\label{BF5}
\sin{\theta}\lr{w^2\p_{\ell}-\nabla\ell}+\cos{\theta}\lr{w^{2}\p_{\psi}-\nabla\psi}=0.
\end{equation}
Projecting \eqref{BF5} on the cotangent basis $\lr{\nabla\ell,\nabla\psi,\nabla\theta}$ gives two linearly independent conditions in $U$:
\begin{subequations}\label{BF6}
\begin{align}
&\cos{\theta}~\sin{\theta}\lr{\abs{\nabla\psi}^2-\abs{\nabla\ell}^2}=\lr{\nabla\ell\cdot\nabla\psi}\lr{\cos^2{\theta}-\sin^{2}{\theta}},\\
&\sin{\theta}~\nabla\ell\cdot\nabla\theta+\cos{\theta}~\nabla\psi\cdot\nabla\theta=0.
\end{align}
\end{subequations}
Thus, we have shown that $\bol{w}$ is a nontrivial Beltrami field provided that we can find a local coordinate
system $\lr{\ell,\psi,\theta}$ satisfying equation \eqref{w2} 
and system \eqref{BF6}. This completes the proof of the first implication.
The proof of the converse statement is immediate, since system \eqref{BF6} guarantees that \eqref{w2} is a Beltrami field, i.e. that it satisfies equation \eqref{Beq}.

\end{proof}
 
Observe that a coordinate system $\lr{\ell,\psi,\theta}$ obeying 
system \eqref{BF6} is very close to an orthogonal coordinate system:
at all points where either $\sin{\theta}$ or $\cos{\theta}$ vanish,
the coordinates $\lr{\ell,\psi,\theta}$ must be orthogonal.
Theorem \ref{LocRep} has the following consequences.
\begin{corollary}\label{Invariants}
Let $\bol{w}\in C^{\infty}\lr{\Omega}$ be a smooth vector field in a bounded domain $\Omega\subset\mathbb{R}^3$ satisfying equation \eqref{Beq} with $h=\bol{w}\cdot\nabla\cp\bol{w}\neq\bol{0}$ in $\Omega$. Then the flow generated by $\bol{w}$ has two local integral invariants $\theta$ and $L_{\theta}=\ell\cos{\theta}-\psi\sin{\theta}$,
\begin{equation}
\bol{w}\cdot\nabla\theta=\bol{w}\cdot\nabla L_{\theta}=0~~~~{\rm in}~~U,
\end{equation}
where $\lr{\ell,\psi,\theta}\in C^{\infty}\lr{U}$ is the local coordinate system given in Theorem \ref{LocRep}.
\end{corollary}  

\begin{proof}
The hypothesis of Theorem \ref{LocRep} are verified. Then $\bol{w}$
admits the local representation of equation \eqref{w2}. From \eqref{Beq} it follows that
\begin{equation}
\hat{h}\,\bol{w}=\cu{\bol{w}}=\nabla\theta\cp\nabla\lr{\ell\cos{\theta}-\psi\sin{\theta}}~~~~{\rm in}~~U.\label{hw}
\end{equation}
where $\hat{h}=h/w^2$. Hence,
\begin{equation}
\bol{w}\cdot\nabla\theta=\bol{w}\cdot\nabla L_{\theta}=0~~~~{\rm in}~~U.
\end{equation}
\end{proof}
\noindent Observe that the level sets of the invariant $\theta$ describe the local planes where the flow generated by $\bol{w}$ lies, while the invariant $L_{\theta}=\nabla\theta\cdot\tilde{\bol{L}}$ reflects conservation of the angular momentum-like quantity $\tilde{\bol{L}}=\tilde{\bol{r}}\cp\bol{w}$, $\tilde{\bol{r}}=h^{-1}\lr{\ell\,\nabla\ell+\psi\,\nabla\psi+\theta\,\nabla\theta}$, in the direction $\nabla\theta$ across such planes. 
Note that $\tilde{\bol{L}}$ reduces to the standard angular momentum if $\lr{\ell,\psi,\theta}$ is a Cartesian coordinate system.
Furthermore, since a nontrivial Beltrami flow is locally endowed with two integral invariants, the evolution
of an observable $f\in C^{\infty}\lr{U}$ with respect to such flow can be expressed in terms of a Nambu bracket \cite{Nambu},
\begin{equation}
\dot{f}=\left\{f,\theta,L_{\theta}\right\}=\hat{h}^{-1}\nabla f\cdot\nabla\theta\cp\nabla L_{\theta}~~~~{\rm in}~~U.
\end{equation}
If $\nabla\cdot\bol{w}=0$, one has $\nabla\hat{h}\cdot\bol{w}=0$.
Therefore, from \eqref{hw} we see that the proportionality factor $\hat{h}$ is a local function of the invariants $\theta$ and $L_{\theta}=\ell\cos{\theta}-\psi\sin{\theta}$ alone,
\begin{equation}
\hat{h}=\hat{h}\lr{\theta,L_{\theta}}~~~~{\rm in}~~U.
\end{equation}

\section{Construction of Beltrami fields}
As a consequence of Theorem \ref{LocRep} it is possible to establish a 
method (corollary \ref{Construction}) to construct Beltrami fields with given proportionality coefficient $\hat{h}$. In this section we prove corollary \ref{Construction} and discuss the construction procedure in detail.

\begin{proof}
First observe that by hypothesis the variables $\lr{\ell,\psi,\theta}$ define an orthogonal coordinate system in $\Omega$. Therefore $\abs{\nabla\theta}=\abs{\alpha}\neq 0$ in $\Omega$, implying $\alpha\neq 0$. It follows that the sign of $\alpha$ does not change, and we can write either $\alpha=\abs{\alpha}$ or $\alpha=-\abs{\alpha}$ in $\Omega$. Next, note that the orthogonality condition demands that:
\begin{equation}
\nabla\ell\cdot\nabla\psi=\nabla\ell\cdot\nabla\theta=\nabla\psi\cdot\nabla\theta=0~~~~{\rm in}~~U.\label{ortho}
\end{equation}
Combining \eqref{ortho} with the assumption $\abs{\nabla\ell}=\abs{\nabla\psi}$, we see that system \eqref{eq2} is satisfied. Then, from Theorem \ref{LocRep}, the vector field $\bol{w}=\cos{\theta}~\nabla\psi+\sin{\theta}~\nabla\ell$ is a Beltrami field. Let us evaluate the proportionality factor $\hat{h}$. From the definition,
\begin{equation}
\hat{h}=\frac{\JI{w}}{w^2}=\frac{\nabla\ell\cdot\nabla\psi\cp\nabla\theta}{\abs{\cos{\theta}~\nabla\psi+\sin{\theta}~\nabla\ell}^2}=\sigma \abs{\nabla\theta}=\sigma\abs{\alpha}~~~~{\rm in}~~U,
\end{equation}  
where we used the orthogonality of the local coordinate system and the assumption $\abs{\nabla\ell}=\abs{\nabla\psi}$. Here $\sigma=h/\abs{h}$ is the sign of the helicity density $h=\JI{w}=\nabla\ell\cdot\nabla\psi\cp\nabla\theta$.

Consider now the vector field $\bol{w}^{\ast}=\sin{\theta}~\nabla\psi+\cos{\theta}~\nabla\ell$. We have
\begin{equation}
\begin{split}
\cu{\bol{w}^{\ast}}&=-h\lr{\cos{\theta}~\frac{\nabla\psi\cp\nabla\theta}{h}+\sin{\theta}~\frac{\nabla\theta\cp\nabla\ell}{h}}
\\
&=-h\lr{\cos{\theta}~\frac{\nabla\ell}{\abs{\nabla\ell}}\abs{\p_{\ell}}+\sin{\theta}~\frac{\nabla\psi}{\abs{\nabla\psi}}\abs{\p_{\psi}}}=-\hat{h}\,\bol{w}^{\ast}~~~~{\rm in}~~U.
\end{split}
\end{equation}
Here we used equation \eqref{tvectors}, the orthogonality condition, the assumption $\abs{\nabla\psi}=\abs{\nabla\ell}$, and the fact that $\hat{h}=h/w^2=h/\abs{\nabla\psi}^2$. This concludes the proof.
\end{proof}

In some cases it is useful to express the vector field $\bol{w}$ in
terms of the tangent basis. From equation \eqref{BF5}, 
\begin{equation}
\bol{w}=\cos{\theta}~\nabla\psi+\sin{\theta}~\nabla\ell=w^{2}\lr{\cos{\theta}~\p_\psi+\sin{\theta}~\p_\ell}.
\end{equation}
If the coordinate system $\lr{\ell,\psi,\theta}$ is orthogonal 
with $\abs{\nabla\ell}=\abs{\nabla\psi}$ we have
\begin{equation}
\bol{w}=\cos{\theta}~\nabla\psi+\sin{\theta}~\nabla\ell=
\abs{\nabla\psi}^2\lr{\cos{\theta}~\p_{\psi}+\sin{\theta}~\p_{\ell}}.\label{wt}
\end{equation}
By using equation \eqref{wt}, we can evaluate the divergence $\nabla\cdot\bol{w}$ in the coordinate system $\lr{\ell,\psi,\theta}$. First recall that
\begin{equation}
dx\w dy\w dz=h^{-1}\,d\ell\w d\psi \w d\theta.
\end{equation}
Hence,
\begin{equation}
\lr{\nabla\cdot\bol{w}}\,dx\w dy\,\w dz\,=\mf{L}_{\bol{w}}\lr{h^{-1}\,d\ell\w d\psi\w d\theta}.
\end{equation}
It follows that
\begin{equation}
\nabla\cdot\bol{w}=\frac{\p w_{x}}{\p x}+\frac{\p w_{y}}{\p y}+\frac{\p w_{z}}{\p z}=h\,\left[\frac{\p}{\p\ell}\lr{w^2\frac{\sin{\theta}}{h}}+\frac{\p}{\p\psi}\lr{w^2\frac{\cos{\theta}}{h}}\right].
\end{equation}
If the coordinate system $\lr{\ell,\psi,\theta}$ is orthogonal 
with $\abs{\nabla\ell}=\abs{\nabla\psi}$ this gives
\begin{equation}
\nabla\cdot\bol{w}=\frac{\p w_{x}}{\p x}+\frac{\p w_{y}}{\p y}+\frac{\p w_{z}}{\p z}=\abs{\nabla\psi}^2\abs{\nabla\theta}\,\left[\frac{\p}{\p\ell}\lr{\frac{\sin{\theta}}{\abs{\nabla\theta}}}+\frac{\p}{\p\psi}\lr{\frac{\cos{\theta}}{\abs{\nabla\theta}}}\right].
\end{equation}

Corollary \ref{Construction} can be slightly generalized as follows:
\begin{proposition}
Let $\lr{\ell,\psi,\theta}\in C^{1}\lr{\Omega}$ be an orthogonal coordinate system in a bounded domain $\Omega\in\mathbb{R}^3$ such that, in $\Omega$,
\begin{subequations}
\begin{align}
\abs{\nabla\theta}&=\lvert\alpha\rvert,\label{Eik}\\
\abs{\nabla\ell}&=\abs{\nabla\psi},
\end{align}
\end{subequations} 
where $\alpha\in C\lr{\Omega}$. 
Let $f\in C\lr{\ov{\Omega}}$ be a function of the variable $\theta$, $f=f\lr{\theta}$. Then, the vector fields
\begin{subequations}
\begin{align}
&\bol{w}=\cos\lr{\int{f\,d\theta}}~\nabla\psi+\sin\lr{\int{f\,d\theta}}~\nabla\ell,\\
&\bol{w}^{\ast}=\sin\lr{\int{f\,d\theta}}~\nabla\psi+\cos\lr{\int{f\,d\theta}}~\nabla\ell.
\end{align}
\end{subequations}
are Beltrami fields in $\Omega$ with proportionality factors $\sigma\lvert\alpha\rvert\,f$ and $-\sigma\lvert\alpha\rvert\,f$ respectively, where $\sigma$ is the sign of the Jacobian $\nabla\ell\cdot\nabla\psi\cp\nabla\theta$. 
\end{proposition}
A related result is the following.
\begin{proposition}
Let $\lr{\alpha,\beta,\gamma}\in C^{1}\lr{\Omega}$ be an orthogonal coordinate system in a bounded domain $\Omega\in\mathbb{R}^3$ with $\abs{\nabla\alpha}=\abs{\nabla\beta}$.
Let $f\in C^{1}\lr{\Omega}$ be a function of the variable $\gamma$, $f=f\lr{\gamma}$. Then the vector field
\begin{equation}
\bol{w}=\frac{1}{\sqrt{1+f^{2}}}~\nabla\beta+\frac{f}{\sqrt{1+f^2}}~\nabla\alpha,
\end{equation} 
is a Beltrami field with proportionality factor $\sigma\abs{\nabla\gamma} f_{\gamma}/\sqrt{1+f^2}$, where $\sigma$ is the sign of the Jacobian $\nabla\alpha\cdot\nabla\beta\cp\nabla\gamma$.
\end{proposition} 
\noindent The proof of this statement follows by performing the change of variable $\tau=\arctan{f}$.

Let us explain how corollary \ref{Construction} is applied. 
Suppose that we wish to construct a Beltrami field $\bol{w}$ with given
proportionality factor $\alpha\neq 0$, 
i.e. such that $\cu{\bol{w}}=\alpha\,\bol{w}$.
This can be accomplished by solving the following system of first order partial differential equations for the variables $\lr{\ell,\psi,\theta}$ in the domain $\Omega$:
\begin{subequations}\label{BFC}
\begin{align}
&\abs{\nabla\theta}=\abs{\alpha},\label{Eik}\\
&\abs{\nabla\ell}=\abs{\nabla\psi},\\
&\nabla\ell\cdot\nabla\psi=0,\\
&\nabla\ell\cdot\nabla\theta=0,\\
&\nabla\psi\cdot\nabla\theta=0.
\end{align}
\end{subequations}
If a solution $\lr{\ell,\psi,\theta}\in C^{1}\lr{\Omega}$ is found, we set $h=\nabla\ell\cdot\nabla\psi\cp\nabla\theta$. 
Then the desired Beltrami field is $\bol{w}=\cos{\theta}~\nabla\psi+\sin{\theta}~\nabla\ell$ for $\alpha >0$ and $h>0$ or $\alpha <0$ and $h<0$, and
$\bol{w}^{\ast}=\sin{\theta}~\nabla\psi+\cos{\theta}~\nabla\ell$ for $\alpha<0$ and $h>0$ or $\alpha >0$ and $h<0$. A Beltrami field constructed in this way is
endowed with the invariants $\theta$ and $L_{\theta}=\ell\cos{\theta}-\psi\sin{\theta}$ of corollary \ref{Invariants} in the whole $\Omega$.

Equation \eqref{Eik} is the eikonal equation and can be solved
by application of the method of characteristics \cite{Evans2}.
However, notice that, even if the variable $\theta$ is known, 
the existence of the orthogonal coordinate system $\lr{\ell,\psi,\theta}$ is not guaranteed due to the coupling among the remaining equations in system \eqref{BFC}. 

If a solenoidal Beltrami field is needed, system \eqref{BFC} must be supplied with an additional equation arising from the condition $\nabla\cdot\bol{w}=0$. Specifically, the obtained vector fields $\bol{w}$ and $\bol{w}^{\ast}$ will be solenoidal provided that
\begin{subequations}\label{sol}
\begin{align}
&\nabla\cdot\bol{w}=\cos{\theta}~\Delta\psi+\sin{\theta}~\Delta\ell=0~~~~{\rm in}~~\Omega,\\
&\nabla\cdot\bol{w}^{\ast}=\sin{\theta}~\Delta\psi+\cos{\theta}~\Delta\ell=0~~~~{\rm in}~~\Omega.
\end{align}
\end{subequations}   
Note that these equations are satisfied simultaneously if the coordinates $\ell$ and $\psi$ are harmonic, i.e. if $\Delta\ell=\Delta\psi=0$ in $\Omega$.

System \eqref{BFC} combined with one of the equations in \eqref{sol}  
provides a method to produce solutions of system \eqref{IE2}, which describes steady incompressible fluid flow. This approach can be generalized to the compressible case. To see this we consider the steady compressible ideal Euler equations
\begin{equation}\label{IE4}
\lr{\bol{w}\cdot\nabla}\bol{w}=-\rho^{-1}\nabla P,~~~~\nabla\cdot\lr{\rho\,\bol{w}}=0~~~~{\rm in}~~\Omega.
\end{equation}
Here $\rho>0$ represents fluid density. 
In the following we assume that both $\rho$ and $P$ are smooth in $\Omega$. Assuming a barotropic pressure $P=P\lr{\rho}$, we can write $\nabla P/\rho=\nabla\phi$ for some appropriate function $\phi=\phi\lr{\rho}$. Then system \eqref{IE4} reduces to
\begin{equation}\label{IE5}
\bol{w}\cp\lr{\nabla\cp\bol{w}}=\nabla\lr{\phi+\frac{w^2}{2}},~~~~\nabla\cdot\lr{\rho\,\bol{w}}=0~~~~{\rm in}~~\Omega.
\end{equation}
We look for a solution of \eqref{IE5} in terms of a nontrivial Beltrami field. 
Then $\phi+w^2/2=c$, with $c$ a real constant. If the function $\phi$ is invertible, we can write $\rho=\phi^{-1}\lr{c-w^2/2}$, and system \eqref{IE5} reduces to
\begin{equation}\label{IE6}
\bol{w}\cp\lr{\nabla\cp\bol{w}}=\bol{0},~~~~\nabla\log\left[{\phi^{-1}\lr{c-\frac{w^2}{2}}}\right]\cdot\bol{w}+\nabla\cdot\bol{w}=0~~~~{\rm in}~~\Omega.
\end{equation}
For an ideal gas $P=k\,\rho$, with $k$ a positive real constant. This gives $\phi=k\log{\rho}$ and system \eqref{IE6} becomes
\begin{equation}\label{IE7}
\bol{w}\cp\lr{\nabla\cp\bol{w}}=\bol{0},~~~~\nabla\lr{w^2}\cdot\bol{w}-2k~\nabla\cdot\bol{w}=0~~~~{\rm in}~~\Omega.
\end{equation}

Equation \eqref{IE5} can be cast in an equivalent form by observing that if $\bol{w}$ is a nontrivial Beltrami field $\rho\,\bol{w}=\rho\,\hat{h}^{-1}\nabla\cp\bol{w}$. Then:
\begin{equation}\label{IE8}
\bol{w}\cp\lr{\nabla\cp\bol{w}}=\bol{0},~~~~\nabla\left\{\left[{\phi^{-1}\lr{c-\frac{w^2}{2}}}\right]\hat{h}^{-1}\right\}\cdot\nabla\cp\bol{w}=0~~~~{\rm in}~~\Omega.
\end{equation}
In local coordinates $\lr{\ell,\psi,\theta}$ the continuity equation is thus equivalent to
\begin{equation}\label{IE9}
\begin{split}
\nabla\left\{\left[{\phi^{-1}\lr{c-\frac{\lr{\cos{\theta}~\nabla\psi+\sin{\theta}~\nabla\ell}^2}{2}}}\right]
\frac{\lr{\cos{\theta}~\nabla\psi+\sin{\theta}~\nabla\ell}^2}{\nabla\ell\cdot\nabla\psi\cp\nabla\theta}\right\}\\\cdot\nabla\theta\cp\nabla \lr{\ell\cos{\theta}-\psi\sin{\theta}}=0~~~~{\rm in}~~U.
\end{split}
\end{equation} 
This shows that the existence of a nontrivial Beltrami field solution to the compressible Euler equations with barotropic pressure \eqref{IE6} is locally equivalent to the existence of 
a coordinate system $\lr{\ell,\psi,\theta}$ such that equations \eqref{eq2}, \eqref{eq3}, and \eqref{IE9} hold.

Equation \eqref{IE9} can be simplified if the coordinate system $\lr{\ell,\psi,\theta}$ is orthogonal (and thus $\abs{\nabla\psi}=\abs{\nabla\ell}$ from system \eqref{eq2}) and $P=k\,\rho$ is the equation of state of an ideal gas. In such case the continuity equation becomes
\begin{equation}
\nabla\lr{\exp\left\{k^{-1}\lr{c-\frac{\abs{\nabla\psi}^2}{2}}\right\}\abs{\nabla\theta}^{-1}}\cdot\nabla\theta\cp\nabla L_{\theta}=0.
\end{equation}
This equation is satisfied provided that there exists a function $u=u\lr{\theta,L_{\theta}}$ such that
\begin{equation}
\abs{\nabla\theta}=\exp\left\{-\frac{\abs{\nabla\psi}^2}{2k}\right\}~u.\label{add}
\end{equation}
It follows that a solution to \eqref{IE7} with given proportionality coefficient $\alpha$ can be obtained by solving system \eqref{BFC} together with the additional condition \eqref{add} with respect to the variables $\lr{\ell,\psi,\theta,u}$.


\section{Examples}
In this section we present a list of Beltrami fields obtained by solving system \eqref{BFC}. Both solenoidal and non-solenoidal Beltrami fields are given.
\begin{enumerate}

\item Let $\lr{r,z,\theta}=\lr{\sqrt{x^2+y^2},z,\arctan\lr{y/x}}$ be a cylindrical coordinate system. The vector field
\begin{equation}
\bol{w}=\cos{z}~\nabla\log{r}+\sin{z}~\nabla\theta,\label{ex1}
\end{equation}
is a Beltrami field with proportionality factor $\hat{h}=-1$.
Moreover $\nabla\cdot\bol{w}=0$.
The flow generated by \eqref{ex1} admits the invariants $z$ and $L_{z}=\theta\,\cos{z}-\log{r}\,\sin{z}$.

\item Consider again the cylindrical coordinate system of the previous example. The vector field
\begin{equation}
\bol{w}=\sin{\theta}~\nabla z+\cos{\theta}~\nabla r,\label{ex2}
\end{equation}
is a Beltrami field with proportionality factor $\hat{h}=1/r$. 
Moreover $\nabla\cdot\bol{w}=\cos{\theta}/r$. The flow generated by \eqref{ex2}
admits the invariants $\theta$ and $L_{\theta}=z\,\cos{\theta}-r\,\sin{\theta}$.

\item Let $\lr{u,v,z}=\lr{\sqrt{r+x},\sqrt{r-x},z}$, $r=\sqrt{x^2+y^2}$, be a parabolic cylindrical coordinate system. This coordinate system is orthogonal with $\abs{\nabla u}=\abs{\nabla v}=1/\sqrt{2r}$. The vector field
\begin{equation}
\bol{w}=\cos{z}~\nabla v+\sin{z}~\nabla u,\label{ex3}
\end{equation}
is a Beltrami field with proportionality factor $\hat{h}=1$. Moreover $\nabla\cdot\bol{w}=0$. The flow generated by \eqref{ex3} admits the invariants $z$ and $L_{z}=u\,\cos{z}-v\,\sin{z}$.

\item Let $\lr{\xi,\eta,\theta}=\lr{\sqrt{\rho+z},\sqrt{\rho-z},\arctan\lr{y/x}}$, $\rho=\sqrt{x^2+y^2+z^2}$, be a parabolic coordinate system. This coordinate system is orthogonal with $\abs{\nabla \xi}=\abs{\nabla \eta}=1/\sqrt{2\rho}$.
The vector field
\begin{equation}
\bol{w}=\cos{\theta}~\nabla \eta+\sin{\theta}~\nabla \xi,\label{ex4}
\end{equation} 
is a Beltrami field with proportionality factor $\hat{h}=1/r$.
Moreover $\nabla\cdot\bol{w}=\lr{x/\eta+y/\xi}/(2r\rho)$. The flow generated by \eqref{ex4} admits the invariants $\theta$ and $L_{\theta}=\xi\,\cos{\theta}-\eta\,\sin{\theta}$.
\item We want to construct a Beltrami field with proportionality factor $\hat{h}=\exp\lr{x+y}$. Solving the eikonal equation $\abs{\nabla\theta}=\exp\lr{x+y}$ gives
\begin{equation}
\theta=\frac{\exp\lr{x+y}}{\sqrt{2}}.
\end{equation} 
Next, we look for coordinates $\ell$ and $\psi$ such that $\nabla\ell\cdot\nabla\lr{x+y}=0$, $\nabla\psi\cdot\nabla\lr{x+y}=0$, $\nabla\ell\cdot\nabla\psi=0$, and $\abs{\nabla\ell}=\abs{\nabla\psi}$.
These conditions can be satisfied by choosing $\ell=z$ and $\psi=\lr{x-y}/\sqrt{2}$. The desired Beltrami field is thus
\begin{equation}
\bol{w}=\frac{1}{\sqrt{2}}\cos\left[\frac{\exp\lr{x+y}}{\sqrt{2}}\right]~\nabla\lr{x-y}+\sin\left[\frac{\exp\lr{x+y}}{\sqrt{2}}\right]~\nabla z.\label{ex5}
\end{equation}
Moreover $\nabla\cdot\bol{w}=0$. The flow generated by \eqref{ex5} admits the invariants $\theta$ and $L_{\theta}=\ell\,\cos{\theta}-\psi\,\sin{\theta}$.
Figure \ref{fig2} shows a plot of \eqref{ex5} over the integral surfaces $\theta=\rm constant$ and $L_{\theta}=\rm constant$.
\begin{figure}[h]
\hspace*{-0cm}\centering
\includegraphics[scale=0.3]{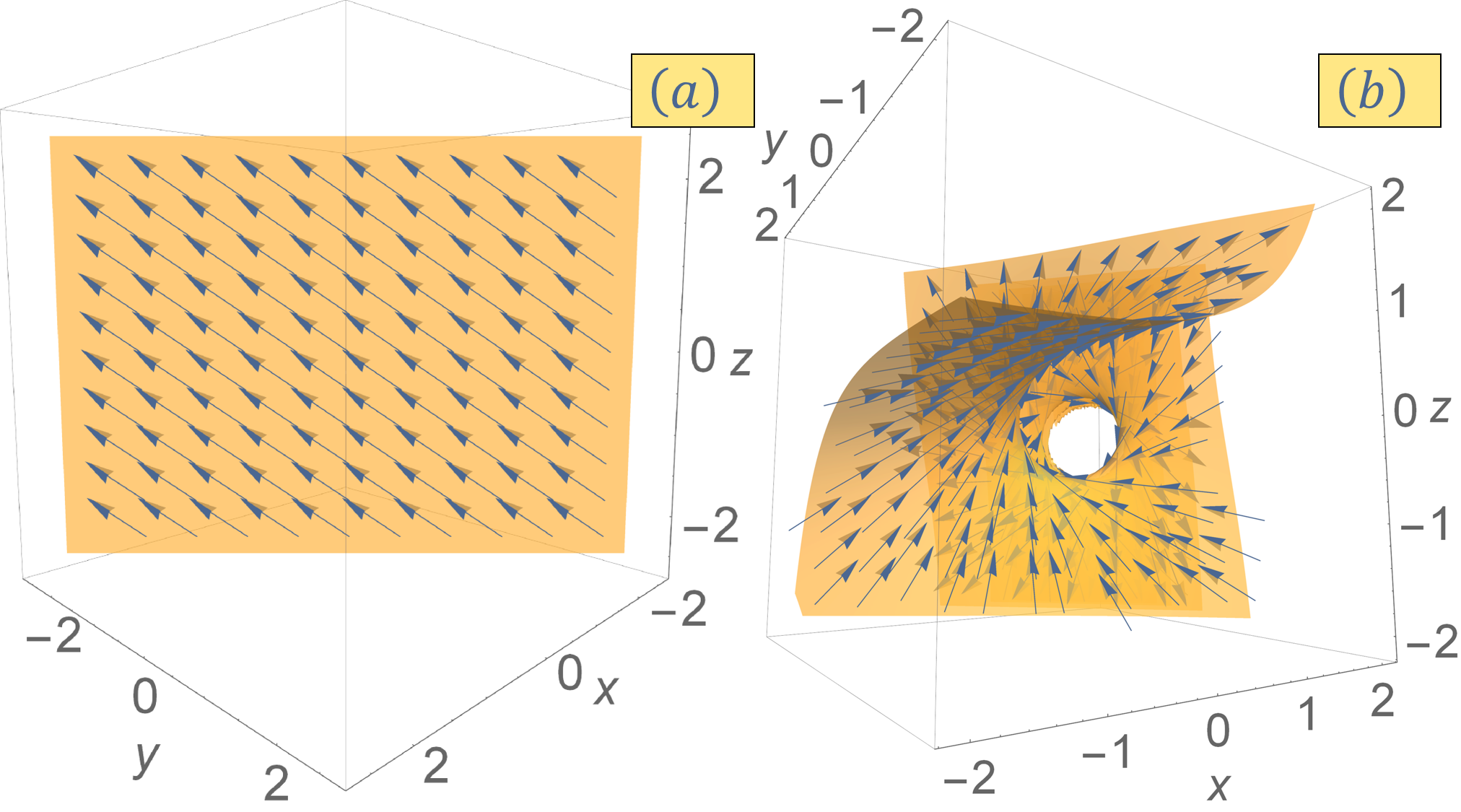}
\caption{\footnotesize (a): Plot of \eqref{ex5} on the surfaces $\theta=0.6$. (b): Plot of \eqref{ex5} on the surface $L_{\theta}=0.5$.}
\label{fig2}
\end{figure}	
 
\item We want to construct a Beltrami field with proportionality factor $\hat{h}=-\cos\lr{x-y}$. Solving the eikonal equation $\abs{\nabla\theta}=\abs{\cos\lr{x-y}}$ gives
\begin{equation}
\theta=\frac{\sin\lr{x-y}}{\sqrt{2}}.
\end{equation} 
Next, we look for coordinates $\ell$ and $\psi$ such that $\nabla\ell\cdot\nabla\lr{x-y}=0$, $\nabla\psi\cdot\nabla\lr{x-y}=0$, $\nabla\ell\cdot\nabla\psi=0$, and $\abs{\nabla\ell}=\abs{\nabla\psi}$.
These conditions can be satisfied by choosing $\ell=\lr{z-x-y}/\sqrt{3}$ and $\psi=\lr{x+y+2z}/\sqrt{6}$. The desired Beltrami field is thus
\begin{equation}
\begin{split}
\bol{w}=\cos\left[{\frac{\sin\lr{x-y}}{\sqrt{2}}}\right]\nabla\lr{\frac{x+y+2z}{\sqrt{6}}}+\\\sin\left[\frac{\sin\lr{x-y}}{\sqrt{2}}\right]\nabla\lr{\frac{z-x-y}{\sqrt{3}}}.\label{ex6}
\end{split}
\end{equation}
Moreover $\nabla\cdot\bol{w}=0$. The flow generated by \eqref{ex6} admits the invariants $\theta$ and $L_{\theta}=\ell\,\cos{\theta}-\psi\,\sin{\theta}$.
Figure \ref{fig3} shows a plot of \eqref{ex6} over the integral surfaces $\theta=\rm constant$ and $L_{\theta}=\rm constant$.
\begin{figure}[h]
\hspace*{-0cm}\centering
\includegraphics[scale=0.3]{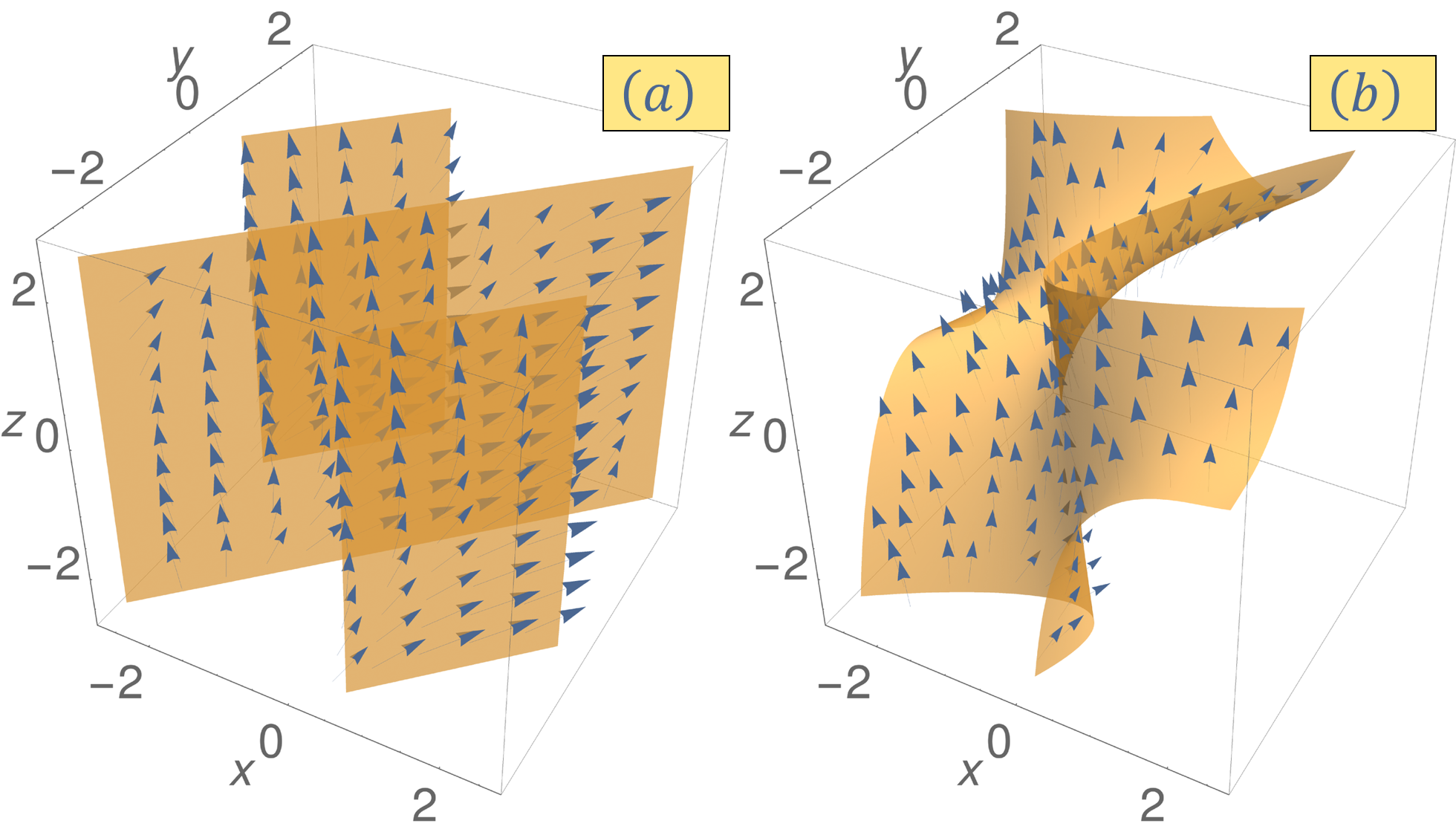}
\caption{\footnotesize (a): Plot of \eqref{ex6} on the surfaces $\theta=0$. (b): Plot of \eqref{ex6} on the surface $L_{\theta}=1$.}
\label{fig3}
\end{figure}
\item We want to construct a Beltrami field with proportionality factor $\hat{h}=\arctan\lr{x+y+z}$. Solving $\abs{\nabla\theta}=\abs{\arctan\lr{x+y+z}}$, we obtain
\begin{equation}
\theta=\left\{\lr{x+y+z}\arctan\lr{x+y+z}-\frac{1}{2}\log\left[{1+\lr{x+y+z}^2}\right]\right\}/\sqrt{3}.
\end{equation} 
The remaining coordinates can be chosen to be $\ell=\lr{x-y}/\sqrt{2}$, and $\psi=\lr{x+y-2z}/\sqrt{6}$. The desired Beltrami field is thus
\begin{equation}
\begin{split}
\bol{w}=\cos\left\{\frac{\lr{x+y+z}\arctan\lr{x+y+z}-\frac{1}{2}\log\left[{1+\lr{x+y+z}^2}\right]}{\sqrt{3}}\right\}\\\nabla\left[{\frac{{x+y-2z}}{\sqrt{6}}}\right]+\\\sin\left\{\frac{\lr{x+y+z}\arctan\lr{x+y+z}-\frac{1}{2}\log\left[{1+\lr{x+y+z}^2}\right]}{\sqrt{3}}\right\}\\\nabla\lr{\frac{x-y}{\sqrt{2}}}.\label{ex7}
\end{split}
\end{equation}
Moreover $\nabla\cdot\bol{w}=0$. The flow generated by \eqref{ex7} admits the invariants $\theta$ and $L_{\theta}=\ell\,\cos{\theta}-\psi\,\sin{\theta}$.
Figure \ref{fig4} shows a plot of \eqref{ex7} over the integral surfaces $\theta=\rm constant$ and $L_{\theta}=\rm constant$.
\begin{figure}[h]
\hspace*{-0cm}\centering
\includegraphics[scale=0.3]{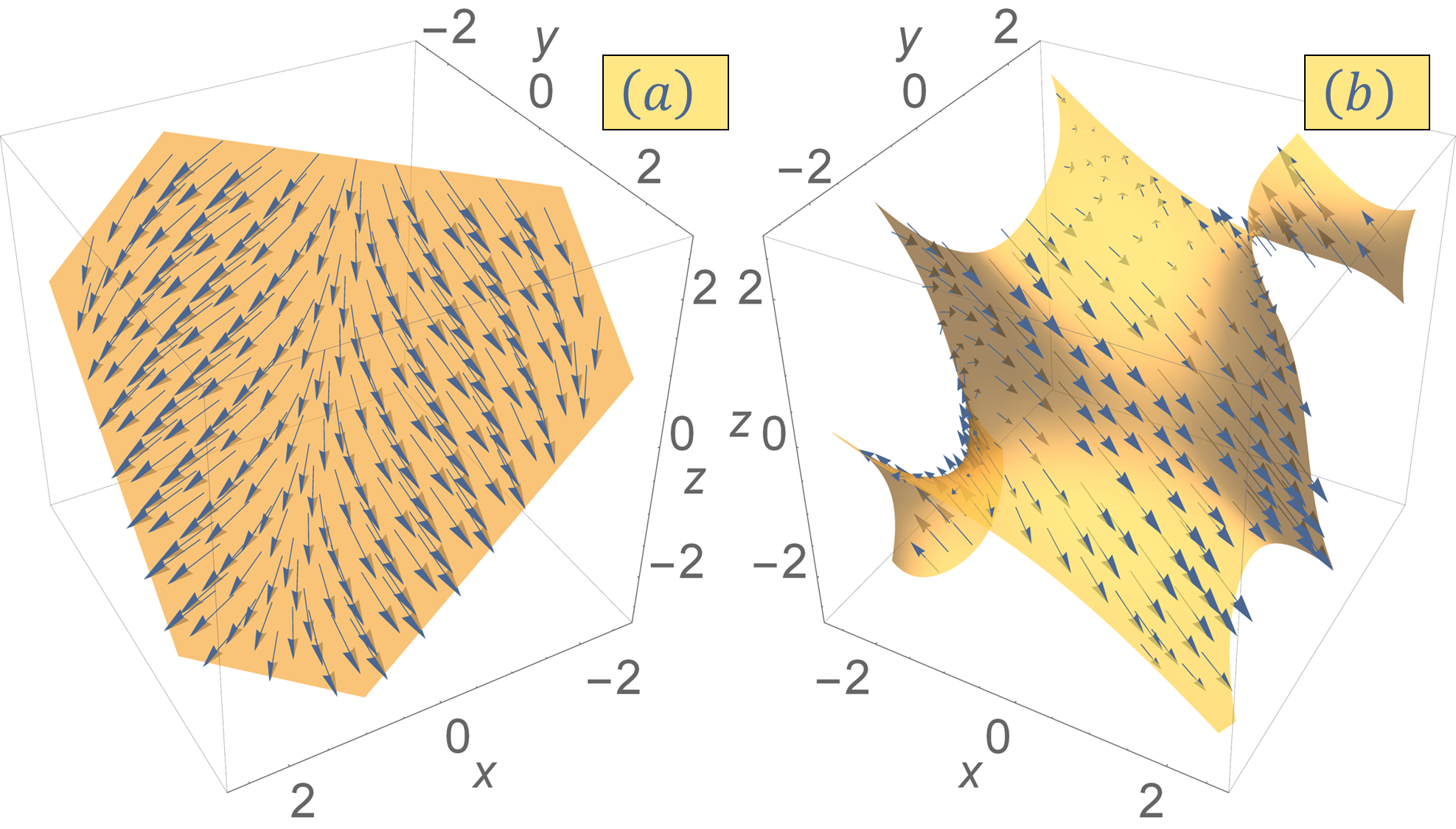}
\caption{\footnotesize (a): Plot of \eqref{ex7} on the surface $\theta=0.2$. (b): Plot of \eqref{ex7} on the surface $L_{\theta}=0.5$.}
\label{fig4}
\end{figure}

\item We look for an orthogonal coordinate system $\lr{\ell,\psi,\theta}$ such that $\ell=\ell\lr{x,y}$, $\psi=\psi\lr{x,y}$, $\theta=z$, and $\abs{\nabla\psi}=\abs{\nabla\ell}$. These conditions are satisfied provided that
\begin{equation}
\ell_{y}=\pm\psi_{x},~~~~\ell_{x}=\mp\psi_{y}.\label{lxly}
\end{equation}
Differentiating each equation with respect to $x$ and $y$, we have
\begin{equation}
\ell_{xx}+\ell_{yy}=\psi_{xx}+\psi_{yy}=0.
\end{equation}
Hence, $\ell$ and $\psi$ must be two-dimensional harmonic functions.
We set $\ell=e^{x}\sin{y}$. Then, integrating \eqref{lxly}, we find $\psi=-e^{x}\cos{y}$.
It follows that the vector field
\begin{equation}
\bol{w}=-\cos{z}~\nabla\lr{e^{x}\cos{y}}+\sin{z}~\nabla\lr{e^{x}\sin{y}},\label{ex8}
\end{equation}
is a Beltrami field with proportionality factor $\hat{h}=1$. Moreover $\nabla\cdot\bol{w}=0$.
The flow generated by \eqref{ex8} admits the invariants $\theta$ and $L_{\theta}=\ell\,\cos{\theta}-\psi\,\sin{\theta}$. Figure \ref{fig5} shows a plot of \eqref{ex8} overt the integral surfaces $\theta={\rm constant}$ and $L_{\theta}={\rm constant}$.

\begin{figure}[h]
\hspace*{-0cm}\centering
\includegraphics[scale=0.3]{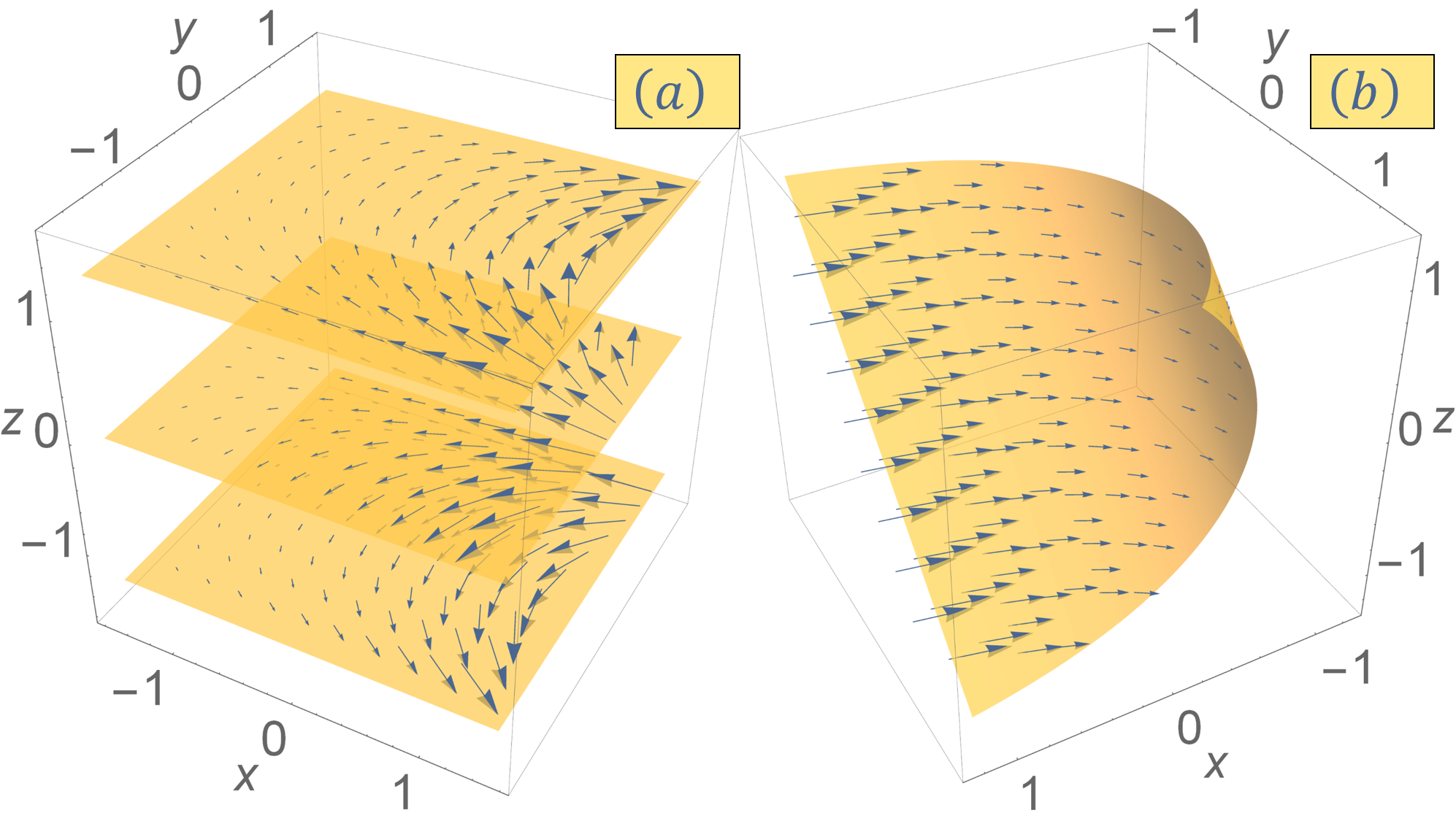}
\caption{\footnotesize (a): Plot of \eqref{ex8} on the surface $\theta=-1.3,0,1.3$. (b): Plot of \eqref{ex8} on the surface $L_{\theta}=0.6$.}
\label{fig5}
\end{figure}

\end{enumerate}

\section{Concluding remarks}

In this study, we formulated a local representation theorem for Beltrami fields
by application of the Lie-Darboux theorem of differential geometry.
We found that, locally, a Beltrami field has a standard form
closely resembling an ABC flow with two of the parameters set to zero. 
In addition, the flow generated by a Beltrami field
is endowed with two local integral invariants that physically
represent the plane of the flow and conservation of an angular momentum-like quantity
in the direction across to the plane.
The obtained local representation naturally leads to
a method to construct Beltrami fields with given proportionality factor:
the solution of the Beltrami field equation \eqref{Beq} can be reformulated into the problem of deriving an 
orthogonal coordinate system satisfying certain geometric conditions.
First we have to solve the eikonal equation, where
the length of one of the cotangent vectors must equal the
proportionality factor. The desired Beltrami field can then
be obtained if the cotangent vector can be 
completed to an orthogonal coordinate system such that the length 
of the two remaining cotangent vectors are equal. 
Using the derived method, we constructed several Beltrami fields
with constant and non-constant proportionality factors, and with
zero and finite divergence. 

\section{Acknowledgments}

\noindent The research of N. S. was supported by JSPS KAKENHI Grant No. 18J01729, 
and that of M. Y. by JSPS KAKENHI Grant No. 17H02860.


\end{document}